\documentclass{amsart}
\usepackage[utf8]{inputenc}
\usepackage{amssymb}
\usepackage{amsmath}
\usepackage{enumitem}
\usepackage{proof}
\usepackage{mathrsfs}
\usepackage{cmll}
\usepackage{xspace}
\usepackage[all,cmtip]{xy}
\usepackage{url}
\usepackage{stmaryrd}

\renewcommand{\emptyset}{\varnothing}

\renewcommand{\alpha}{\propto}

\newcommand{\at}[1]{{\rm #1}}

\newcommand{\argument}[1]{\ensuremath{\mathcal{#1}}}
\newcommand{\prf}[1]{\ensuremath{\mathcal{#1}}}
\newcommand{\proves}[1][]{\vdash}

\newcommand{\base}[1]{\ensuremath{\mathscr{#1}}}
\newcommand{\event}[1]{\base{#1}}

\newcommand{\partialto}{\rightharpoonup}

\newcommand{\seq}{\triangleright}
\newcommand{\ergo}{\textsf{e}}

\newcommand{\calculus}[1]{\ensuremath{\mathsf{#1}}}
\newcommand{\system}[1]{\ensuremath{\mathfrak{#1}}}

\usepackage{xcolor}

\newcommand*{\setStyle}[1]{\ensuremath{\mathbb{#1}}}

\newcommand{\setGoals}{\setStyle{GOALS}\xspace}
\newcommand{\setFormulas}{\setStyle{F}}
\newcommand{\setEvents}{\setStyle{EVENTS}\xspace}

\newcommand{\setAtoms}{\setStyle{P}\xspace}

\newcommand{\setSequent}{\setStyle{S}\xspace}

\newcommand{\funcList}[1]{\textsc{list}(#1)}

\newcommand*{\rn}[1]  {\ensuremath{\mathsf{#1}}}
\newcommand*{\cut}{\rn{cut}}
\newcommand*{\weak}{\rn{w}}
\newcommand*{\exch}{\rn{e}}

\newcommand*{\cont}{\rn{c}}
\newcommand*{\rrn}[2][]  {\rn{#2_{R#1}}}
\newcommand*{\lrn}[2][]  {\rn{#2_{L#1}}}
\newcommand*{\irn}[2][]  {\rn{#2_{I#1}}}
\newcommand*{\ern}[2][]  {\rn{#2_{E#1}}}
\newcommand{\ax}{\mathsf{ax}}

\title{Proof-theoretic Semantics and Tactical Proof}

\author{Alexander V. Gheorghiu} 
\address{University College London} 
\email{alexander.gheorghiu.19@ucl.ac.uk}

\author{David J. Pym}
\address{University College London \& Institute of Philosophy, University of London} 
\email{d.pym@ucl.ac.uk}

\makeatletter
\@ifundefined{thechapter}{\newtheorem{Thm}{Theorem}}{\newtheorem{Thm}{Theorem}[chapter]}
  \newtheorem{Theorem}[Thm]{Theorem}   
\newtheorem{Definition}[Thm]{Definition}      %                   "
\newtheorem{Example}[Thm]{Example}               %  
\newtheorem{Proposition}[Thm]{Proposition}
\newtheorem{Corollary}[Thm]{Corollary}
\makeatother

\begin{document}

% \author{\small 
% \begin{tabular}{cc}
% Alexander V. Gheorghiu & David J. Pym \\ 
% University College London & University College London \& \\ 
% London WC1E 6BT, UK    & Institute of Philosophy, University of London \\ 
%     & London WC1E 6BT, UK \\[1.5ex] 
% alexander.gheorghiu.19@ucl.ac.uk & d.pym@ucl.ac.uk 
% \end{tabular}
% }

\date{} 

\maketitle

\begin{abstract}
The use of logical systems for problem-solving may be as diverse as in proving theorems in mathematics or in figuring out how to meet up with a friend. In either case, the problem solving activity is captured by the search for an \emph{argument}, broadly conceived as a certificate for a solution to the problem. Crucially, for such a certificate to be a solution, it has be \emph{valid}, and what makes it valid is that they are well-constructed according to a notion of inference for the underlying logical system. We provide a general framework uniformly describing the use of logic as a mathematics of reasoning in the above sense. We use proof-theoretic validity in the Dummett-Prawitz tradition to define validity of arguments, and use the theory of tactical proof to relate arguments, inference, and search.  
\end{abstract}

% \keywords{logical systems, tactics, tactical proof, proof-search, proof-theoretic semantics}

\section{Introduction}\label{sec:intro}
The definition of a system of logic may be given \emph{proof-theoretically} as a collection of rules of inference that, when composed in specified ways, determine \emph{proofs}; that is, formal constructions that establish 
that a conclusion is a consequence of some assumptions or axioms. In other words, proofs are objects regulated by rules of a system that determine the inference of a 
conclusion from premisses that have been established: 
\[
\frac{\mathrm{Established \; Premiss}_1 \quad \ldots \quad \mathrm{Established \; Premiss}_k}{\mathrm{Conclusion}}{\big\Downarrow}
\]
We call this proof-theoretic formulation \emph{deductive logic}. 

Deductive logic is useful as a way of defining what proofs are, but it does not reflect either how logic is typically used in practical reasoning problems or the method by which proofs are found. In practice, proofs are typically constructed by starting with a desired, or putative, conclusion and applying the rules of inference \emph{backward}. Read from conclusion to premisses, the rules are sometimes called 
\emph{reduction operators}, and denoted
\[
\frac{\mathrm{Sufficient \; Premiss}_1 \quad \ldots \quad \mathrm{Sufficient \; Premiss}_k}{\mathrm{Putative \; Conclusion}} {\big\Uparrow}  
\]
We call the constructions in a system of reduction operators \emph{reductions}. Crucially, the space of reductions contains proofs and objects that are not proofs --- specifically, reductions that cannot be continued in a way that eventually reaches a proof.
This proof-theoretic formulation has been dubbed \emph{reductive logic}~\cite{pym2004reductive}. 

Reductive logic more closely resemble what mathematicians do when proving theorems and, more generally, how people solve problems using formal representations. It is also the paradigm on logic used for diverse applications in informatics and other systems-oriented sciences, including, but not limited to, areas such as program and systems verification, natural language processing, and knowledge representation. Indeed, it is the reductive perspective that underpins the use of logic in proof assistants --- for example, LCF, HOL, Isabelle, Coq, Mizar, Twelf, and more~\cite{wikiproofassistants}. More generally, it is the paradigm capturing the deployment of logic in the context of capturing human reasoning, which is discussed extensively in, for example, the work of Kowalski~\cite{Kowalski1986} and Bundy~\cite{Bundy1983}. 

There are semantics of proofs, which give abstract accounts of object that certify a statement as being valid in a logic, with a prime example given by the categorical treatment of the BHK interpretation of intuitionistic logic. But, importantly, reductive logic is concerned not only with what proofs but more generally with things that may be continued to form proofs through reduction. The semantics of proofs do not treat this bigger space of objects, but these are central in the use of logic as a mathematics of reasoning.   Therefore, given the widespread application of reductive logic as a reasoning technology, we require a semantic framework explicating the meaning of reductions in the context of reasoning. Reciprocally, we require a general framework stating what \emph{valid} reasoning is relative to a certain problem domain. Such a framework would provide a basis for understanding and analyzing the aforementioned wide-ranging applications of logic in practical reasoning tasks, which is analogous to the use of semantics for programming languages. 

A general framework supporting the mechanization of reductive logic is the theory of \emph{tactical proof} introduced by Milner~\cite{milner1984tactics}. While little developed mathematically, the theory is sufficiently general to encompass as diverse reasoning activities as proving a formula in a formal system and seeking to meet a friend before noon on Saturday. It does not concern finding the \emph{best} way to reason about a goal (e.g., minimizing the prospect of failure), though these things are important, instead it makes precise how concepts used during reasoning --- such as `goal,' `strategy,' `achievement,' `failure,' etc. ---relate to one another. Crucially, the theory of tactical proof is a very general meta-theoretic framework that subsumes specific procedures 
such as focusing, uniform proof-search, resolution, matrix methods, and so on. It is tactical proof that delivers systematically the various proof-assistants mentioned above.

In this paper, we use  the theory of  tactical proof to give a general account of the relationship between arguments --- that is, abstract entities representing reasoning --- and inference in a logic. In so doing, we provide a semantic framework for logic in practical reasoning problems, as witnessed by tactical proof. Importantly, we do not say that this framework is how one \emph{should} go about using logic as a mathematics of reasoning, but rather we aim to \emph{describe} how logic is typically used in the literature. We justify the framework by some coherence theorems relating arguments, reductive reasoning (witnessed by tactics), and inference, and then give a series of examples where the framework is implicit in the literature. While tactical proof connects inference and arguments, the semantics of reasoning is given in the paradigm of \emph{proof-theoretic semantics} (P-tS).

Introduced by Dummett~\cite{dummett1991logical} (based on results by Prawitz~\cite{prawitz1965}) and developed by Schroeder-Heister~\cite{schroeder2006validity,SEP-PtS}, proof-theoretic validity is a semantics of derivations (i.e., objects constructed inductively using rules of a formal system)  in Gentzen's \calculus{NJ}~\cite{Gentzen}. A derivation in \calculus{NJ} is (proof-theoretically) valid if and only if every closure ---  the result of substituting its open assumption for proofs --- of it reduces \emph{\'a la} Prawitz~\cite{prawitz1965} to a normal proof. The priority of normal proofs is justified according to the principle that the introduction rules of Gentzen's \calculus{NJ} act as definitions of the logical constants they introduce --- see Schroeder-Heister~\cite{SEP-PtS} for a discussion. The details are given in Section~\ref{sec:ptv-IPL}. In this paper, we generalize the notion of proof-theoretic validity to accommodate the full potential of the theory of tactical proof.

The philosophy on which P-tS is based is  \emph{inferentialism}, the view that rules of inference 
confer meaning to expressions --- see Brandom~\cite{Brandom2000}). This background renders P-tS the appropriate way (i.e., in contrast to, for example, model-theoretic semantics) for understanding reductive logic and reasoning because it ensures that the meaning of a reduction depends on the accepted use within the activity of reasoning.

The paper begins in Section~\ref{sec:logic} with a terse but complete introduction to the setup of logic needed in this this paper. It continues with a definition of proof-theoretic validity for intuitionistic propositional logic, which motivates later work. Following this, we present the theory of tactical proof and discuss its application to logic. In Section~\ref{sec:semantics-of-proofs}, we briefly explain what semantics of proofs have so far achieved and the way in which they are limited for understanding the use of logic as a reasoning technology. In Section~\ref{sec:main}, we propose a semantic framework that consider the entire space of reductions, which proceeds through a general account of proof-theoretic validity. We justify the framework by a correctness theorem and through a series of examples from the literature in which it is implicit. The paper concludes in Section~\ref{sec:conclusion} with a summary of results.

\section{Logic as Consequence} \label{sec:logic}
In this paper, we study logics generally, which is to say that instead of considering a specific logic, we develop a general framework uniformly applicable across logics. This requires us first to define what we mean by \emph{a logic}. Doing this is controversial. We choose to err on the side of being over-encompassing, so almost any logic in the literature is readily captured. 

Moreover, we also require a notion of inference. There are many paradigms we may consider, such as \emph{inter alia}, natural deduction systems, axiomatic systems (i.e., Hilbert calculi), and tableaux systems. The sequent calculus format is sufficiently expressive to capture all these and is, therefore, the paradigm used in this paper. 

\subsection{Consequence} \label{sec:logic:consequence}

A logic is captured by a relation called \emph{consequence} over data structures called \emph{sequents}. We will insist on nothing about these structures, and we do not impose any particular properties of the consequence relation.  Additional restrictions are often helpful in practice but are not required for the general framework presented herein. 

Fix a set $\setSequent$ of data-structures called \emph{sequents}. 

\begin{Definition}[Consequence] \label{def:consequence}
     A consequence relation is a predicate on a set of sequents. 
\end{Definition}

We think of consequence as a judgement asserting that a certain sequent is valid in the logic. Typically, such as in Example~\ref{ex:ipl:sequents}, sequents can be thought of as tuples of data structures so that we may refer to consequence as a relation between these structures.

\begin{Example}[Intuitionistic Sequents] \label{ex:ipl:sequents}
    Fix a set of atomic propositions \setAtoms. The set of formulas $\setFormulas$ (over \setAtoms) is constructed by the following grammar:
    \[
    \phi ::= \rm p \in \setAtoms \mid \phi \lor \phi \mid \phi \land \phi \mid \phi \to \phi \mid \bot 
    \]
    An intuitionistic sequent is a pair $\Gamma \seq \phi$ in which $\Gamma$ is a list of formulas and $\phi$ is a formula. 
\end{Example}

We have presented a definition of consequence that is, perhaps, overly generous. Nonetheless, this definition means we readily capture any logic within the literature. Refining the definition is challenging because logics may have quite diverse data-structures, ranging from the familiar list/multiset/set contexts of classical and intuitionistic logics, to the structurally complex bunched contexts of relevance logics, and beyond, rendering refinements quite subtle in order to not be arbitrarily exclusive. The more substantial point is that we do not require such refinement. Hence, restricting our notion of consequence according to some ideological principle would be excessive and needless. Crucially, the definition of consequence coheres naturally with the level of generality offered by the theory of tactical proof (see Section~\ref{sec:tactics}), which is the framework we use in this paper to capture the deployment of logic as a reasoning technology.

It is all very well to have a logic. It remains to define what reasoning steps the logic supports. In this paper, what we mean by reasoning step is \emph{inference} in a calculus for consequence.

\subsection{Sequent Calculi} \label{sec:logic:sequentcalculus}

Taking the perspective on logic in Section~\ref{sec:logic}, an inference is the process of beginning with some sequents --- thought of as a putative consequence of a logic --- and ending with another sequent --- thought of as being entailed, according to our logic, by the original sequents. These inferences are understood as instances of rules. 

\begin{Definition}[Rule] \label{def:rule}
     A rule is a relation $\rn{r}$ on sequents.
\end{Definition}

That $\rn{r}(s,s_1,\ldots ,s_n)$ obtains may be denoted by inference schemas:
\[
%\infer{s}{s_1 & \ldots & s_n}
\dfrac{s_1 \ldots s_n}{s} \, r
\]
The sequent $s$ is said to be \emph{conclusion} and the sequents $s_1,\ldots ,s_n$ the \emph{premisses}. A rule that has no premisses (i.e., a predicate on sequents) is called an \emph{axiom}-rule.

With this notion of rule, we give the standard treatment of sequent calculi proofs --- see, for example, Troestra and Schwichtenberg~\cite{troelstra2000basic}.

\begin{Definition}[Calculus] \label{def:calculus}
     A calculus is a set of rules containing at least one axiom-rule.
\end{Definition}

\begin{Definition}[Proof] \label{def:proof}
 Let $\calculus{L}$ be a calculus. The set of $\calculus{L}$-proofs is the set of rooted trees of sequents inductively constructed as follows:
 
 \textsc{Base Case.} If there is an axiom-rule $\rn a \in \calculus{L}$ such that $\rn a(s)$, then the tree of just the node $s$ is an $\calculus{L}$-proof.
 
 \textsc{Induction Step.} If $\prf{P}_1,\ldots ,\prf{P}_n$ are $\calculus{L}$-proofs with roots $s_1,\ldots ,s_n$, respectively, and there is a sequent $s$ and a rule $\rn{r} \in \calculus{L}$ such that $\rn{r}(s,s_1,\ldots .,s_n)$ obtains, then the argument $\prf{P}$ with root $s$ and immediate sub-trees $\prf{P}_1,\ldots ,\prf{P}_n$ is a proof.
\end{Definition}

The notion of proof from a calculus induces a notion of provability from the calculus, which is a consequence relation: 

\begin{Definition}[Provability] \label{def:provability}
    A sequent is \calculus{L}-provable iff there is a \calculus{L}-proof that concludes it.
\end{Definition}

We say that a calculus $\calculus{L}$ \emph{characterizes} a logic when $\calculus{L}$-provability coincides with the consequence relation of the logic. To be precise, this coincidence has two directions: \emph{soundness} and \emph{completeness}. The calculus is \emph{sound} for the logic when it only proves the consequences of the logic, and it is \emph{complete} when it can prove all of the consequences of the logic. 

\begin{Definition}[Soundness and Completeness] \label{def:snc}
     Let $\proves$ be a consequence relation and $\calculus{L}$ be a calculus. 
     \begin{itemize}
         \item[-] Calculus $\calculus{L}$ is \emph{sound} for $\proves$ iff, for any $s \in \setSequent$, if $\proves_{\calculus{L}} s$, then $ \proves s$.
         \item[-] Calculus $\calculus{L}$ is \emph{complete} for $\proves$ iff, for any $s \in \setSequent$, if $\proves s$, then $ \proves_{\calculus{L}} s$.
     \end{itemize}
\end{Definition}

Of course, several different calculi may characterize a logic, some of which may differ substantially. One way to generate new calculi from old is by including rules that are conservative over provability. Such rules are said to be \emph{admissible} --- that is, $\rn{r}$ is admissible in $\calculus{L}$ iff $\proves_{\calculus{L}\cup\{\rn{r}\}}$ is sound with respect to $\proves_{\calculus{L}}$. More generally, a rule can be admissible for a logic when it preserves validity in the logic --- that is, that is, $\rn{r}$ is admissible for $\proves$ iff, for any $s, s_1,...,s_n \in \setSequent$, if $\rn{r}(s,s_1,...,s_n)$ obtains and $\proves s_1$, $\hdots$ , $\proves s_n$, then $\proves s$. Observe that if $\calculus{L}$ is sound for $\proves$, then it is necessarily the case that all the rules in $\calculus{L}$ are admissible for $\proves$. 

\begin{Example}[Example~\ref{ex:ipl:sequents} cont'd] \label{ex:ipl:lj}
    Calculus $\calculus{LJ}$ is given by the rules of Figure~\ref{fig:LJ}, in which $\Gamma$ and $\Gamma'$ are understood to be permutations of each other as lists. This calculus characterizes intuitionistic logic (IPL) --- that is, an intuitionistic sequent $\Gamma \seq \phi$ is a consequence in IPL iff it is $\calculus{LJ}$-provable. Gentzen~\cite{Gentzen} proved that the following rules is admissible for IPL:
    \[
    \infer[\rn{cut}]{\Delta, \Gamma \seq \chi}{\Delta \seq \phi & \phi, \Gamma \seq \chi}
    \]
\end{Example}

\begin{figure}[t]
    \fbox{
	\begin{minipage}{.95\textwidth}
		\centering 
		\vspace{0.2cm}
		$
		\infer[\lrn{\weak}]{\phi,\Gamma \seq \Delta}{\Gamma\seq\Delta}
		\quad
		\infer[\rrn{\weak}]{\Gamma \seq \phi}{\Gamma\seq\emptyset}
		\quad
		\infer[\lrn{\cont}]{\phi,\Gamma \seq \Delta}{\phi,\phi,\Gamma\seq\Delta}
		\quad
		\infer[\exch]{\Gamma \seq \Delta}{\Gamma'\seq \Delta'}
		$
		\\[1.5ex]
		$
		\infer[\rrn\land]{\Gamma \seq \phi \land \psi}{\Gamma \seq \phi & \Gamma\seq\psi}
		\quad
		\infer[\lrn{\land^1}]{\phi \land  \psi,\Gamma\seq \Delta}{\phi,\Gamma\seq\Delta}
		\quad
		\infer[\lrn{\land^2}]{\phi\land \psi,\Gamma\seq \Delta}{\psi,\Gamma\seq\Delta}
		\quad
				\infer[\rrn \neg]{\Gamma \seq \neg \phi}{\phi,\Gamma \seq \emptyset}
		$
		\\[1.5ex]
		$
		\infer[\lrn{\lor}]{\phi\lor\psi,\Gamma \seq \Delta}{\phi,\Gamma \seq\Delta & \psi,\Gamma \seq\Delta }
	       \quad
		\infer[\rrn{\lor^1}]{\Gamma \seq \phi \lor  \psi}{\Gamma \seq \phi }
		\quad
		\infer[\rrn{\lor^2}]{\Gamma \seq \phi \lor  \psi}{\Gamma \seq\psi }
		\quad
				\infer[\lrn \neg]{\neg \phi , \Gamma \seq \emptyset}{\Gamma \seq \phi}
			$
		\\[1.5ex]
		$
		\quad
		\infer[\rrn \to]{\Gamma \seq \phi \to \psi}{\phi,\Gamma \seq \psi}
		\quad
		\infer[\lrn \to]{\phi \to \psi,\Gamma_1,\Gamma_2 \seq \Delta}{\Gamma_1 \seq \phi & \psi,\Gamma_2 \seq \Delta}
        \quad
		\infer[\ax]{\phi \seq \phi}{}
		$
    \end{minipage}
    }
    \caption{System \calculus{LJ}}
    \label{fig:LJ}
\end{figure}

This concludes a terse but complete account of a general perspective of logic and its proof-theoretic formulation as required in this paper. 

Our position is that the objects constructed in various paradigms of argumentation (e.g., natural deduction, dialogue games, etc.) are meaningful precisely to the extent that they represent reasoning in a logic. In this paper, we take `reasoning in a logic' that to mean proof in a sequent calculus characterizing a logic. For example, consider the tree-like constructions from Gentzen's \calculus{NJ} (see Section~\ref{sec:ptv-IPL}), which are supposed to represent intuitionistic reasoning --- indeed, we may say that they are a \emph{natural} representation of reasoning (see, for example, Tennant~\cite{Tennant78entailment}). We take the view that what makes them valid in intuitionistic logic is their relationship to intuitionistic consequence, as captured by the relationship to Gentzen's \calculus{LJ}. We use the theory of tactical proof to make this idea of the relationship between reasoning and logic precise. 

\section{Proof-theoretic Validity for Intuitionistic Propositional Logic} \label{sec:ptv-IPL}

Central to this paper is the paradigm of \emph{proof-theoretic semantics} (P-tS). We defer an explanation of the relevance of P-tS to the problem of this paper to Section~\ref{sec:semantics-of-proofs}, and presently give an account of it for intuitionistic propositional logic (IPL). This will facilitate the more general discussion  of the subject below and motivate later work.

\subsubsection{Intuitionistic Propositional Logic} \label{sec:ptv-IPL:IPL}

We begin by fixing the terminology for IPL as used in this paper; this extends the setup of Example~\ref{ex:ipl:sequents} and Example~\ref{ex:ipl:lj}.

\begin{Definition}[IPL Formula]
     Fix a set of atomic propositions \setAtoms. The set of IPL formulas $\setFormulas$ (over \setAtoms) is constructed by the following grammar:
    \[
    \phi ::= \at{p} \in \setAtoms \mid \phi \lor \phi \mid \phi \land \phi \mid \phi \to \phi \mid \bot 
    \]
\end{Definition}

We use the following abbreviations:
\[
\hat{\Gamma} := \bigwedge_{\phi \in \Gamma} \phi \qquad \neg \phi := \phi \to \bot
\]

\begin{Definition}[IPL Sequent]
     An IPL sequent is a pair $\Gamma \seq\phi$ in which $\Gamma$ is a set of formulas and $\phi$ is a formula.
\end{Definition}

We will characterize the consequence relation for IPL by proof in Gentzen's \calculus{NJ}~\cite{Gentzen}. To this end, we introduce some terminology for the paradigm of argument known as \emph{natural deduction}.

\begin{Definition}[Natural Deduction Argument]
    A natural deduction argument is a rooted tree of formulas in which some (possibly no) leaves are marked as discharged.

    An argument is open if it has undischarged assumptions; otherwise, it is closed.
\end{Definition}

In this section we will simply say \emph{argument} to mean natural deduction argument. The leaves of an argument are its \emph{assumptions}, the root is its \emph{conclusion}. That $\argument{A}$ has open assumptions $\Gamma$, closed assumptions $\Delta$, and conclusion $\phi$ may be denoted as follows:
\[
\deduce{\phi}{\argument{A}} \qquad \deduce{\argument{A}}{\Gamma,[\Delta]} \qquad \deduce{\phi}{\deduce{\argument{A}}{\Gamma,[\Delta]}}
\]

Intuitively, an argument $\argument{A}$ with open assumptions $\Gamma$ and conclusion $\phi$, is an argument \emph{for} the sequent $\Gamma \seq\phi$. We call this sequent the consequence of the argument. We use a system of natural deduction to define a class of arguments whose consequence precisely coincide with consequences of IPL. 

A natural deduction system is given by rules governing the composition of arguments. We follow the standard treatments
(see, for example, van Dalen~\cite{vanDalen} or Troelstra and Schwichtenberg~\cite{troelstra2000basic}), doubtless already familiar.  

\begin{Definition}[Natural Deduction System \calculus{NJ}]
    The natural deduction system $\calculus{NJ}$ is composed of the rules in Figure~\ref{fig:nj}.
\end{Definition}

\begin{figure}[t]
   \fbox{
   \begin{minipage}{0.95\linewidth}
   \centering
   $
       \infer[\irn{\land}]{\phi \land \psi}{\phi & \psi} \qquad \infer[\ern{\land^1}]{\phi}{\phi \land \psi} \quad \infer[]{\psi}{\phi \land \psi}
       \qquad \infer[\irn{\to}]{\phi \to \psi}{\deduce{\phi}{[\psi]}} 
       \qquad
       \infer[\ern{\bot}]{\phi}{\bot} 
    $
    \\[1.5ex]
    $
    \infer[\irn{\lor^1}]{\phi \lor \psi}{\phi} 
    \quad 
    \infer[\irn{\lor^2}]{\phi \lor \psi}{\psi} 
    \qquad
    \infer[\ern{\lor}]{\chi}{\phi \lor \psi & \deduce{\chi}{[\phi]} & \deduce{\chi}{[\psi]}}
    \qquad
    \infer[\ern{\to}]{\phi}{\phi & \phi \to \psi} 
    $
   \end{minipage}
   }
    \caption{Calculus $\calculus{NJ}$}
    \label{fig:nj}
\end{figure}

\begin{Definition}[\calculus{NJ}-Derivation]
    The set of \calculus{NJ}-derivations is defined inductively as follows:
    
    \textsc{Base Case.} If $\phi$ is a formula, then the one element tree $\phi$ is a \calculus{NJ}-derivation.
    
    \textsc{Inductive Step.} Let $\rn{r}$ be a rule in $\calculus{NJ}$ and $\argument{D}_1,...,\argument{D}_n$ be \calculus{NJ}-derivations. If $\mathcal{D}$ is an argument arising from applying $\rn{r}$ to $\mathcal{D}_1,...,\mathcal{D}_n$, then $\mathcal{D}$ is an \calculus{NJ}-derivation.
\end{Definition}

A \emph{closed \calculus{NJ}-derivation} in a calculus is an \calculus{NJ}-\emph{proof}. The existence of such proofs characterizes IPL.

\begin{Proposition}[Gentzen~\cite{Gentzen}]
\label{prop:nj}
     There is an \calculus{NJ}-derivation of $\Gamma \seq \phi$ iff $\Gamma \proves \phi$.
\end{Proposition}

To give an account of proof-theoretic validity for IPL in the Dummett-Prawitz tradition, we require some auxiliary definitions.

The rules of $\calculus{NJ}$ with subscripts $\rn{I}$ and $\rn{E}$ are the \emph{introduction rules} ($I$-rules) and \emph{elimination rules} ($E$-rules), respectively.

\begin{Definition}[Detour]
    A \emph{detour} in a derivation is a sub-derivation in which a formula is obtained by an $I$-rule and is then the major premise of the corresponding $E$-rule.
\end{Definition} 
\begin{Example}
The following is a detour for conjunction:
\[
\infer[\ern{\to}]{\psi}{
    \deduce{\phi}{\mathcal{D}_1} & 
    \infer[\irn{\to}]{\phi \to \psi}
        {
            \deduce{\psi}{\deduce{\mathcal{D}_2}{[\phi]}}
        }
    }
\]
\end{Example}
\begin{Definition}[Canonical Derivation]
    A derivation is canonical iff it contains no detours.
\end{Definition}

Prawitz~\cite{prawitz1965} proved that canonical \calculus{NJ}-proofs are complete for IPL. The argument uses a reduction relation $\rightsquigarrow$ that precisely eliminates detours; for example, detours with implication are reduced as follows:
\[
\infer[\ern{\to}]{\psi}{
    \deduce{\phi}{\mathcal{D}_1} & 
    \infer[\irn{\to}]{\phi \to \psi}
        {
            \deduce{\psi}{\deduce{\mathcal{D}_2}{[\phi]}}
        }
    }
    \qquad
    \raisebox{1em}{$\rightsquigarrow$}
    \qquad
    \deduce{\psi}{
         \deduce{\mathcal{D}_2}{
            \deduce{\phi}{\mathcal{D}_1}
            }
        }
\]
The reflexive and transitive closure of $\rightsquigarrow$ is denoted $\rightsquigarrow^\ast$. This reduction relation is normalizing and its normal forms are canonical derivations.

\begin{Proposition}[Prawitz~\cite{prawitz1965}] \label{prop:normal}
     There is a canonical \calculus{NJ}-derivation of $\Gamma \seq \phi$ iff $\Gamma \proves \phi$.
\end{Proposition}

This establishes the relevant syntax and proof theory required for IPL in this paper. In following subsection, we present proof-theoretic validity for natural deduction.
 
\subsubsection{Proof-theoretic Validity} \label{sec:ptv-IPL:ptv}
 It was Dummett~\cite{dummett1991logical} who first realized the philosophical significance of the normalization result by Prawitz~\cite{prawitz1965} (i.e., Proposition~\ref{prop:normal}). This begins the development of proof-theoretic validity as a semantics of arguments. We follow the account given by Schroeder-Heister~\cite{schroeder2006validity}.

When introducing natural deduction, Gentzen~\cite{Gentzen} offered the following remarks:

 \begin{quote}
The introductions represent, as it were, the ‘definitions’ of the symbols concerned, and the eliminations are no more, in the final analysis, than the consequences of these definitions. This fact may be expressed as follows: In eliminating a symbol, we may use the formula with whose terminal symbol we are dealing only ‘in the sense afforded it by the introduction of that symbol’. 
\end{quote}

The following corollary of the normalization result (Proposition~\ref{prop:normal}) supports this view:
 
\begin{Corollary} \label{cor:intro}
    There is a canonical \calculus{NJ}-derivation of $\Gamma \seq \phi$ concluding with an introduction rule iff $\Gamma \proves \phi$.
\end{Corollary}

Together, the above suggest an inferentialist semantics of arguments according to the following principals:
\begin{itemize}
    \item[-] canonical proofs are \emph{a priori} valid
    \item[-] closed derivations are valid as a consequence of them reducing to canonical proofs
    \item[-] open derivations are regarded as placeholders for closed derivations according to their possible closures.
\end{itemize}

What we are calling proof-theoretic validity is a specific realization of this semantic brief. To definite it, we must define what is meant by the \emph{closure} of a derivation. 

The idea is that we substitute open assumptions for closed derivations. However, an argument may contain open assumptions that are not consequences of IPL and thus (by Proposition~\ref{prop:nj}) do not admit \calculus{NJ}-proofs. How are such things closed? Therefore, to form the semantics, we consider validity in possible extensions of \calculus{NJ} that can handle such assumptions. 

If the open assumptions of an argument are complex formulas, then one may use rules of \calculus{NJ} to produce something simpler. What remains to consider is the case where the open assumptions are atomic propositions. In this case, their meaning is supplied by an atomic \emph{base}. 

\begin{Definition}[Base] \label{def:base}
    A base is a finite set of rules over atomic propositions,
    \[
    \infer{\at{ c}}{\at{ p}_1 & ... & \at{ p}_n}
    \]
    The premisses of the rule may be empty; that is, the rule may be an axiom.
\end{Definition}

The atomic rules are not necessarily closed under substitution. To see how they confer meaning to atomic propositions, consider the statement `Tammy is a vixen'. What does it mean? Intuitively, it means, somehow, that `Tammy is a vixen' \emph{and} `Tammy is a fox'. In the current setup, this is to say that we have base containing the following rules:
\[
{\small 
\begin{array}{c}
\infer{\text{Tammy is a vixen}}{\text{Tammy is a fox} & \text{Tammy is female}}\qquad
\infer{\text{Tammy is female}}{\text{Tammy is a vixen}} \qquad \infer{\text{Tammy is a fox}}{\text{Tammy is a vixen}} \\
\end{array}
}
\]
This is because the sense of `and' in the above heuristic is made precise by the inferential semantics 
 of $\land$, defined by the rules $\irn \land$, $\ern{\land^1}$, and $\ern{\land^2}$, which are emulated by the above rules.

\begin{Definition}[Derivation in a Base] \label{def:bvalid}
    Let $\base{B}$ be a base. The set of $\base{B}$-derivations is inductively defined as follows:
    
    \textsc{Base Case.} If $\at{p} \in \setAtoms$ is an atom, then the one element tree consisting of $\at{ p}$ is a $\base{B}$-derivation.
    
    \textsc{Inductive Step.} Let $\rn{r}$ be a rule in $\base{B}$ and $\argument{D}_1,...,\argument{D}_n$ be \base{B}-derivations. If $\mathcal{D}$ is an argument arising from applying $\rn{r}$ to $\mathcal{D}_1,...,\mathcal{D}_n$, then $\mathcal{D}$ is a \base{B}-derivation.
\end{Definition}

 Relative to a base (including the empty base) one has a simple proof-theoretic semantics for the formulas of IPL: a formula is valid (relative to the base) iff it admits a valid argument (relative to the base). 

\begin{Definition}[Validity Natural Deduciton Argument] \label{def:ptv:ipl}
    Let $\base{B}$ be a base. The $\base{B}$-valid of an argument $\argument{A}$ is inductively defined as follows:
    \begin{itemize}
        \item[-] $\argument{A}$ is a closed $\base{B}$-derivation
        \item[-] $\argument{A}$ is a closed canonical $\calculus{NJ}\cup\base{B}$-derivation whose immediate sub-derivations $\argument{A}_1, ... \argument{A}_n $ are $\base{B}$-valid
        \item[-] $\argument{A}$ is a closed non-canonical $\calculus{NJ}\cup\base{B}$-derivation that reduces to a $\base{B}$-valid canonical derivation $\argument{A}'$
        \item[-] $\argument{A}$ is an open derivation and, for every $\base{C} \supseteq \base{B}$, any extension of $\argument{A}$ by $\base{C}$-valid arguments of the assumptions $\argument{C}_1....,\argument{C}_n$ is a $\base{C}$-valid argument.
    \end{itemize}
    An argument is valid iff it is $\base{B}$-valid for every base $\base{B}$.
\end{Definition}

Schroeder-Heister~\cite{schroeder2006validity} has relativized this definition to the atomic system considered and to the set of justifications (proof reductions) considered. In Section~\ref{sec:main:ptv}, we go further and relativize it to the notion of argument too.

This concludes the definition of proof-theoretic validity for IPL. In the next section, we explain what it has to do with \emph{reasoning} in IPL and the r\^ole of tactics in relating natural deduction to consequence for the logic. 

Intuitively, the rules of $\calculus{NJ}$ are instructions for constructing arguments. When \emph{reasoning} in \calculus{NJ}, we begin with an argument for a putative conclusion and use the rules of \calculus{NJ} backwards to construct a closed \calculus{NJ}-derivation. What makes each reasoning step valid is that an admissible rules witness it. The relationship suggested here between \calculus{NJ}- and $\calculus{LJ}$-proofs justifying this is, of course, well-known --- see, for example, Gentzen~\cite{Gentzen}. It may be summarized as saying that constructions in $\calculus{NJ}$ are intuitionistic \emph{proofs} and constructions in $\calculus{LJ}$ are \emph{constructions of intuitionistic proofs}. We use the theory of tactical proof to make this idea precise.

\section{The Theory of Tactical Proof}  \label{sec:tactics}

The theory of tactical proof is a meta-theoretic framework supporting reductive logic and its mechanization. Indeed, tactical proof is sufficiently general to encompass diverse reasoning activities at various levels of formality, which may be as widely different as proving a theorem in a mathematical theory, seeking to win at chess, and arranging to meet a friend before noon on Saturday. The point is that all such reasoning activities, different in domain and formality, can be articulated in terms of a uniform language that a user may express insight into reasoning methods and delegate routine, but error-prone, work to a machine. In the words of Milner~\cite{milner1984tactics}:

\begin{quote}
    Here it is a matter of taste whether the human prover wishes to see this performance done by the machine, in all its frequently repulsive detail, or wishes only to see the highlights, or is merely content to let the machine announce the result (a theorem!).
\end{quote}

Following Milner~\cite{milner1984tactics}, we introduce the theory at the full level of generality at first and then concentrate on it in the context of proof-search in logics. 

\subsection{Tactics and Tacticals} \label{sec:tactics:milner}

One has two classes of prime entities: \emph{goals} and \emph{events}. The two classes are carried by the sets $\setGoals$ and $\setEvents$, respectively. The goals and events are related by a notion of \emph{achievement} $\alpha \, \subseteq \setGoals \times \setEvents$ that determines what events witness what goals. The idea is that an event $\event{E}$ achieves goal $G$ as it satisfies the description that the goal has designated.Heuristically, an event achieves a goal when it satisfies the description that the goal has designated. For example, the goal $G$ that Alice and Bob meet before noon on Saturday is achieved by the event $\event{E}$ is that Alice and Bob meet under the clock at Waterloo station at 11:53 on Saturday. 

We take reasoning about a goal as the process of replacing it with new goals that suffice to produce the original. In the nomenclature of reductive logic, such replacements are captured by \emph{reduction operators}, which may be taken as a partial function from goals to lists of goals:
\[
\rho:\setGoals \partialto \funcList{\setGoals}
\]
The goals produced by applying a reduction operator to a given goal are said to be \emph{subgoals}. 

What renders a reduction from a goal to a list of subgoals valid is that any events possibly witnessing the subgoals yield an event possibly witnessing the original goal. This justification is witnessed by a \emph{procedure},
\[
\pi:\funcList{\setEvents} \partialto \setEvents
\]
Returning to the example above concerning Alice and Bob, the goal $G$ may be reduced to the following sub-goals:
    \[
    \begin{array}{ccl}
    G_1 & :& \text{Alice arrives under the clock at Waterloo Station before noon on Saturday }\\
    G_2 & : & \text{Bob arrives under the clock at Waterloo Station before noon on Saturday.}
    \end{array}
    \]
This reduction is justified by the fact that $G_1$ and $G_2$ are achieved by the following events, respectively, which yield $\event{E}$ through the procedure of \emph{waiting}:
    \[
      \begin{array}{ccl}
    \event{E}_1 & : & \text{Alice arrives at Waterloo Station at 11:57 on Saturday }\\
    \event{E}_2 & : & \text{Bob arrives at Waterloo Station at 11:53 on Saturday.}
    \end{array}
    \]

Thus, one step of reasoning amounts to applying a (partial) mapping takings goals to subgoals together with a procedure,
\[
\tau: G \mapsto \langle [G_1, \ldots ,G_n], \pi \rangle 
\]
These mappings are called \emph{tactics}.
According to the above discussion, we have the following notion of validity:

\begin{Definition}[Valid Tactic] \label{def:validtactic}
 Let $\alpha$ be a notion of achievement. A tactic $\tau$ is $\alpha$-valid iff, for any $G, G_1, \ldots ,G_n \in \setGoals$ and $\event{E},\event{E}_1,\ldots ,\event{E}_n \in \setEvents$, if $\tau : G \mapsto \langle [G_1, \ldots ,G_n], \pi \rangle$ and $\event{E} := \pi(\event{E}_1, \ldots ,\event{E}_n)$, and $\event{E}_i \propto G_i$ obtains for $1 \leq i \leq n$, then $\event{E} \propto G$ obtains.
\end{Definition}

Of course, a goal typically requires several iterations of reasoning of the above form such that subgoals are resolved into further subgoals, and so on. For example, suppose Alice starts from Andover and Bob starts from Birmingham; then, to reason about $G$, one requires many component tactics that collectively bridge the distance both physical and temporal --- for instance, one may have the subgoal $G_1'$ for $G_1$ that `Bob takes the tube to Waterloo Station from Euston Station', which is witnessed by the event `Bob takes the 11:46 southbound Norther Line service from Euston to Waterloo on Saturday.' Hence, we require a notion of composition of tactics. 

A composition of tactics is called a \emph{tactical}.
A tactical is valid when it preserves the validity of the tactics it combines:

\begin{Definition}[Valid Tactical]
   A tactical is valid iff it preserves the validity of tactics; that is, if $\circ$ is a tactical and $\tau_1$,\ldots ,$\tau_n$ are $\propto$-valid, then $\circ(\tau_1,\ldots ,\tau_n)$ are $\propto$-valid.
\end{Definition}

The foregoing is a complete account of tactical reasoning as introduced by Milner~\cite{milner1984tactics}.  To be precise in the semantics presented in this paper, we supplement the above with some additional definitions. 

\begin{Definition}[Tactical System]
    A tactical system $\calculus{T}$ is a collection of tactics and tacticals that are valid relative to some notion of achievement.
\end{Definition}

We have opted to present the theory in its full generality. In the sequel, we apply it to the context of the use of logic as a reasoning technology. We follow the account in Milner~\cite{milner1984tactics}, which is the basis of many automated reasoning technologies using logic, such as the proof assistants mentioned in Section~\ref{sec:intro}.

    Let $\proves$  be the consequence relation for IPL --- see Section~\ref{sec:ptv-IPL}. We have the following setup:
    \begin{itemize}
        \item[-] a goal is a sequent $\Gamma \seq \phi$ in which $\Gamma$ is a list of formulas and $\phi$ is a formula
        \item[-] an event is a consequence $\Delta \proves \psi$
        \item[-] the achievement relation $\alpha$ is as follows:
        \[
         (\Delta \seq \psi) \propto (\Gamma \seq \phi) \qquad \text{ iff } \qquad \phi = \psi \text{ and } \Delta \sqsubseteq \Gamma \text{ and } \Delta \proves \psi
        \]
       (We write $\Delta \sqsubseteq \Gamma$ to denote that the set of elements in $\Delta$ is a subset of the set of elements of $\Gamma$)
    \end{itemize}
    In this context, a tactic is valid iff it corresponds to an admissible rule for IPL. For example, in $\calculus{NJ}$ the $\irn{\land}$-rule determines the tactic $\tau_{\irn{\land}}$ which has the following components:
    \[
    \underbrace{\infer[\Uparrow]{\Gamma \seq \phi \land \psi}{\Gamma \seq \phi &  \Gamma \seq \psi}}_{\text{reduction operator}} \qquad \underbrace{\infer[\Downarrow]{\Delta_1, \Delta_2 \proves \phi \land \psi}{\Delta_1 \proves \phi & \Delta_2 \proves \psi}}_{\text{procedure}}
    \]

This concludes the overview of the theory of tactical proof as used in this paper. What has not been given here is an account of exactly what \emph{argument} is constructed by a tactic, which we defer to Section~\ref{sec:main}.  

\subsection{Tactical Proof and Intuitionistic Propositional Logic} \label{sec:tactics:IPL}

Having presented the theory of tactics as a metatheoretical framework in which one studies reasoning --- the construction of arguments ---  in its full generality, it is informative to consider how it applies in the concrete setting of natural deduction for IPL.

    In Section~\ref{sec:tactics:milner}, we witnessed the following tactic corresponding to the $\irn \land$ rule:
    \[
     \tau_\land : (\Gamma \seq \phi \land \psi) \mapsto \langle [(\Gamma \seq\phi), (\Gamma \seq\psi)], \rrn \land \rangle
    \]
The analogous treatment of $\irn\to$ yields the following:
    \[
     \tau_\to : (\Gamma \seq\phi \to \psi) \mapsto \left\langle [(\phi,\Gamma \seq\psi)], \rrn \to \right\rangle
     \]
     These individual reasoning steps are combined with a \emph{tactical} $\fatsemi$ that corresponds to the sequential application of rules in natural deduction:
      \[
    \tau_\land\fatsemi\tau_\to: \big(\Gamma \seq\chi\land (\phi \to \psi) \big)\mapsto \langle [ (\Gamma \seq\chi), (\phi,\Gamma \seq\psi) ] , \rrn \land \otimes \rrn \to  \rangle
     \]
The procedure  $\rrn \land \otimes \rrn \to$ is the product of the procedures for $\rrn \land$ and $\rrn \to$.

 We have thus related natural deduction and consequence using tactics. And yet, something is missing in this setup. What argument does $\tau_\land\fatsemi\tau_\to$ witness? This question demands an interpretation of tactics as arguments, understood as abstract entities such as natural deduction arguments. 
 
 Moreover, the justification for proof-theoretic validity came from the idea that the introduction rules are definitional. That is all very well, but it isn't the only choice. Depending on various desiderata, we are left with various answers to the question, \emph{what is a valid argument?} By framing the relationship between arguments and consequence in the theory of tactical proof, we can define various notions of validity on arguments according to the priority, by fiat, of some sequent calculus characterizing consequence and some transformations of arguments. Furthermore, there is nothing about tactics that pertains, in particular, to natural deduction, so the notion of argument herein can be generalized to other paradigms too. These generalizations deliver the semantic framework for logic as a reasoning technology that this paper is about.

\subsection{Tactical Proof and Logic} \label{sec:tactics:logic}

The theory of tactical proof is an engine relating the search for arguments. It is, essentially, a system of reduction operators that are step-wise justified by rules in a sequent calculus. That is how we regard them in this paper. The application of tactics drives the computation of arguments, which is to say the \emph{search} for arguments, and a sequent calculus (in the sense of Section~\ref{sec:logic}) provides procedures. As such, supplying a sequent calculus amounts to supplying a notion of \emph{inference} against which the computation of an argument is justified. This is captured in the semantic framework of this paper in Section~\ref{sec:main}

The dichotomy between proof and search actually predates the tactical proof. The components were historically called \emph{analysis} and \emph{synthesis}, respectively --- see P\'olya~\cite{polya1945solve} for a general discussion of this study for mathematical practice. In \emph{analysis}, one repeatedly asks from what conditions could the desired result, which is to say goal, be obtained; during \emph{synthesis}, one derives from the analysis a solution to the problem. 

In this paper, the shift from analysis to synthesis (i.e., the shift from computing subgoals to using procedures, from reduction to deduction) is captured by a \emph{synthesizer}. 

\begin{Definition}[Synthesizer] \label{def:synthesizer}
Let $\calculus{L}$ be a sequent calculus and $\calculus{T}$ be tactical system with achievement $\alpha$ whose events are $\calculus{L}$-sequents. The achievement $\alpha$ is an $\calculus{L}$-synthesizer for $\calculus{T}$ iff the procedures of $\calculus{T}$ are the rules of $\calculus{L}$.
\end{Definition}

An example of a synthesizer is offered at the end of Section~\ref{sec:tactics:milner}. Here the reduction operators correspond to \calculus{NJ} rules, and the procedures correspond to $\calculus{LJ}$ rules.

In the same way that tactics are implicit in much of the literature on logic, synthesizers also appear implicitly anywhere one considers the inferential content within arguments in a certain space. The running case of natural deduction for IPL discussed above is a key example; we give some others in Section~\ref{sec:main:example}.

\section{The Semantics of Proofs} \label{sec:semantics-of-proofs}

On the one hand, we have structures called arguments that represent evidence for a consequence of a logic --- for example, natural deduction proofs \emph{\`a la} Gentzen~\cite{Gentzen}, \emph{\`a la} 
 Fitch~\cite{fitch1952symbolic}, \emph{\`a la} Lemmon~\cite{lemmon1978beginning},  and so on. On the other hand, we have a formal account of \emph{inference} supplied by fixing a sequent calculus for a logic. We use the theory of tactical proof introduced by Milner~\cite{milner1984tactics} to relate the two. We desire a framework able to explain these connections because reductive logic is central within the practice of logic as a mathematics of reasoning, especially within system-oriented science (e.g., in the use of logic for program and systems verification, natural language processing, and proof assistants). Hence, we require a semantics of proofs to determine \emph{valid} arguments and to consider their construction.

Of course, semantics of proofs for particular logics have been very substantially developed for particular logics. Intuitionism, as defined by Brouwer~\cite{brouwer1913intuitionism}, is the view that an argument is valid when it provides sufficient evidence for its conclusion. This defines intuitionistic logic (IL). A consequence is that IL differs from classical logic by rejecting \emph{tertium non datur}  --- that is, the ability to assert a proposition for the rejection of its negation --- as such an inference does not constitute sufficient evidence for its conclusion. Heyting~\cite{heyting1966intuitionism} and Kolmogorov~\cite{kolmogorov} provided a semantics for intuitionistic proof, which captures the evidential character of intuitionism, called the Brouwer-Heyting-Kolmogorov (BHK) interpretation of IL. It is now the standard explanation of the logic --- see, for example, van Atten~\cite{vanAtten}. Proof-theoretic validity (P-tV) for IPL (see Section~\ref{sec:ptv-IPL}) is, of course, strongly related to BHK --- see Schroeder-Heister~\cite{Schroeder2007modelvsproof}.

The \emph{propositions-as-types} correspondence --- see Howard~\cite{howard1980formulae} --- gives an standard way of instantiating the denotation of proofs in the BHK interpretation of intuitionistic propositional logic (IPL) as terms in the simply-typed $\lambda$-calculus. Technically, the setup can be sketched as follows: a judgement 
% \[
%      \phi_1 , \ldots , \phi_k  \vdash \Phi \,:\, \phi  
% \]
that $\Phi$ is an $\calculus{NJ}$-proof of the sequent $\phi_1 , \ldots , \phi_k \seq \phi$ corresponds to a typing 
judgement 
\[
    x_1 : A_1 , \ldots , x_k : A_k \vdash M(x_1,\ldots, x_k) : A 
\]
where the $A_i$s are types corresponding to the $\phi_i$s, the $x_i$s correspond to placeholders for proofs of the $\phi_i$s, the $\lambda$-term $M(x_1 , \ldots, x_k)$ corresponds to $\Phi$, and 
the type $A$ corresponds to $\phi$. 

Lambek~\cite{lambek1980lambda} gave a more abstract account by showing that simply-typed $\lambda$-calculus is the internal language of cartesian closed categories (CCCs), thereby giving a categorical semantics of proofs for IPL. In this setup, 
a morphism 
\[
    \llbracket \phi_1 \rrbracket \times \ldots \times \llbracket \phi_k \rrbracket 
        \stackrel{\llbracket \Phi \rrbracket}{\longrightarrow} 
            \llbracket \phi \rrbracket
\]
in a CCC (where $\times$ denotes cartesian product) that interprets the $\calculus{NJ}$-proof $\Phi$ of $\phi_1 , \ldots , \phi_k \seq \phi$ also interprets the 
term $M$, where the $\llbracket \phi_i \rrbracket$s interpret also the $A_i$s and $\llbracket \phi \rrbracket$ also interprets 
$A$. 

To generalize to full IL (and beyond), Seely~\cite{seely1983hyperdoctrines} modified this categorical setup and introduced \emph{hyperdoctrines} --- indexed categories of CCCs with coproducts over a base with finite products. Martin-L\"of~\cite{martin1975intuitionistic} gave a formulae-as-types correspondence for predicate logic using dependent type theory. Barendregt~\cite{Barendregt1991} gave a systematic treatment of type systems and the propositions-as-types correspondence. A categorical treatment of dependent types came with Cartmell~\cite{cartmell} --- see also, for examples among many, work by Streicher~\cite{Streicher1988}, Pavlovi\'c~\cite{Pavlovic1990}, Jacobs~\cite{Jacobs}, and Hofmann~\cite{Hofmann1997}. In total, this gives a semantic account of \emph{proof} for first- and higher-order predicate intuitionistic logic based on the BHK interpretation. 

That is all very well as explaining what a proof is for IL, but the space of objects considered when finding an argument also contains things that are not proofs and cannot be continued to form proofs. Constructed by backward inference, we call these objects \emph{reductions}. In IPL, an example of a reduction that fails to be a proof is an \calculus{NJ}-derivation whose open assumptions are not theorems of IPL. While such an argument is well-constructed according to intuitionism, it is not valid since it is not closed and cannot be closed (i.e., the open assumptions cannot be substituted for proofs). We desire a semantics of proofs able to treat the entire space of reductions as the latter are central to the use of logic within practical reasoning. In Section~\ref{sec:ptv-IPL}, we saw that proof-theoretic semantics could account for such arguments, which are beyond the scope of BHK. Pym and Ritter~\cite{pym2004reductive} have provided a general semantics of reductive logic in the context of classical and intuitionistic logic through polynomial categories; that is, by extending the categories in which arrows denote proofs for a logic by additional arrows that supply `proofs' for propositions that do not have proofs but appear during reduction.

This paper aims to provide a uniform framework for describing the relationship between argument, reduction, and inference. Taking the view that what makes an argument valid for a logic is that it respects a notion of inference for that logic, we subscribe to \emph{inferentialism} (see Brandom~\cite{Brandom2000}). We generalize proof-theoretic validity such that we can handle an arbitrary notion of argument (i.e., not just natural deduction \emph{\`a la} Gentzen~\cite{Gentzen}) for an arbitrary logic (i.e., not just IPL). This is the subject of Section~\ref{sec:main}.

\section{Proof-theoretic Semantics and Tactical Proof} \label{sec:main}

In this section, we address the problems raised in Section~\ref{sec:semantics-of-proofs} and present a uniform framework relating arguments, tactical proof (reductive reasoning), and sequent calculi (inference) based on inferentialism. We do this through a generalized view of proof-theoretic validity (see Section~\ref{sec:ptv-IPL}), using the theory of tactical proof as the engine of computation through which inference and argument are related. This work is intended to be descriptive in that it is a uniform metatheoretical platform witnessing how these things are typically related in the literature.

\subsection{Proof-theoretic Validity, generalized} \label{sec:main:ptv}

We begin with a space of \emph{arguments}  $\setStyle{A}$. Within this space, there is a subset $\setStyle{P}\subseteq \setStyle{A}$ arguments that are \emph{a priori} valid; these are called \emph{canonical proofs}. These canonical proofs for the basis on which the validity of all the other arguments is derived. 

A given argument $\argument{A} \in \setStyle{A}$ may represent another argument $\argument{A}' \in \setStyle{A}$ in some way. For example, in the setting of natural deduction (Section~\ref{sec:ptv-IPL}), an argument containing a detour can be thought of as representing a natural deduction argument arising reduction \emph{\`a la} Prawtiz~\cite{prawitz1965}. Thus, we take the space of arguments to be equipped with some \emph{justification operators} of the form $j:\setStyle{A} \partialto \setStyle{A}$ that transform one argument to another. 

It may be that arguments are left \emph{open} in some sense. The idea is that the argument contains all the suasive content they require but have left something unsaid, which can be arbitrarily filled in. Returning to the case of natural deduction (Section~\ref{sec:ptv-IPL}), this was the state of open derivations, which have left the justification of their open assumptions unstated (by the very fact of them being open). Therefore, we further equip the space of arguments with \emph{closure operators} of the form $c: \setStyle{A} \partialto \setStyle{A}$, mapping arguments to arguments.

To summarize:

\begin{Definition}[Argument Space]
     An argument space is a tuple $\system{A} := \langle \setStyle{A},\setStyle{P},\mathcal{J},\mathcal{C} \rangle$ in which $\setStyle{A}$ is the set of arguments, $\setStyle{P}$ is the set of proofs, $\mathcal{J}$ is the set of justification operators $j:\setStyle{A} \partialto \setStyle{A}$, and $\mathcal{C}$ is the set of closure operators $c:\setStyle{A} \partialto \setStyle{A}$. 
\end{Definition}

A notion of proof-theoretic validity precisely analogous to the treatment of IPL in Section~\ref{sec:ptv-IPL} follows immediately:

\begin{Definition}[Proof-theoretic Validity] \label{def:ptv:valid}
     Let $\system{A} := \langle \setStyle{A},\setStyle{P},\mathcal{J},\mathcal{C}\rangle$ be an argument space. An argument $\argument{A}$ is $\system{A}$-valid iff one of the following holds:
     \begin{itemize}
         \item[-] it is a canonical proof ---  $\argument{A} \in \setStyle{P}$;
         \item[-] there is $j \in \mathcal{J}$ such that $j(\argument{A})$ is $\system{A}$-valid;
         \item[-] for any $c \in \mathcal{C}$, the closure $c(\argument{A})$ is \system{A}-valid.
     \end{itemize}
\end{Definition}
\begin{Example}[Proof-theoretic Validity for  IPL] \label{ex:naturalPtS}
  Consider the arguments space $\system{N} := \langle \setStyle{A},\setStyle{D},\setStyle{P},\mathcal{J},\mathcal{C} \rangle$ in which the components are as follows:
    \begin{itemize}
        \item[\setStyle{A} - ] Comprises natural deduction arguments
        \item[\setStyle{P} - ] Comprises canonical $\calculus{NJ}$-proofs
        \item[$\mathcal{J}$ - ] Comprises the reduction transformations by Prawitz~\cite{prawitz1965}
        \item[$\mathcal{C}$ - ] Comprises maps that substitute open assumptions for derivation in a base. 
    \end{itemize}
    The validity  condition from Definition~\ref{def:ptv:valid} instantiated to $\system{N}$ is precisely Definition~\ref{def:ptv:ipl}.
 \end{Example}

Of course, the point of the generalization of proof-theoretic validity in this section is that other examples may be captured too. 
    
\begin{Example}[Proof-search Games] \label{ex:games}
    Consider the game-semantics of proof-search for IPL by Pym and Ritter~\cite{pym2005games}; see also work by Miller and Saurin \cite{miller2006game}. Succinctly, a (partial) strategy is a (partial) function that extends plays --- sequences of moves  --- that end on an opponent move, which satisfy certain conditions. Each strategy represents an attempt at proof-search in \calculus{LJ}. A \emph{winning} strategy (i.e., a strategy satisfying certain conditions) represent a successful proof-search --- that is, proof-search that actually finds an \calculus{LJ}-proof --- but it may include backtracking. We have the argument space $\system{S} = \langle \setStyle{A}, \setStyle{D}, \setStyle{P}, \mathcal{J}, \mathcal{C} \rangle$ in which the components are as follows: 
    \begin{itemize}
        \item[$\setStyle{A}$ - ] Arguments are partial strategies
        \item[$\setStyle{P}$ - ] Canonical proofs are winning strategies without backtracking
        \item[$\mathcal{J}$ - ] The justification operators collapse backtracking sections of strategies
        \item[$\mathcal{C}$ - ] The closure operators extends partial strategies to total strategies.
    \end{itemize}
    The validity condition from Definition~\ref{def:ptv:valid} instantiated to $\mathcal{L}$ renders a strategy valid when it represents an \calculus{LJ}-proof. 
\end{Example}

% Different argument spaces may be superficially different in that what they take for arguments, justification operators, etc., is completely analogous. For example, \calculus{NJ} can be presented using natural deduction arguments as in Section~\ref{sec:ptvIPL}, or it may be presented in sequent calculus form. This merits a notion of \emph{homormorphism} between these spaces. 

% \begin{Definition}[Argument Space Homormorphism] \label{def:homomorphism}
%      Let $\mathcal{L}_1 := \langle \setStyle{A}_1,\setStyle{D}_1,\setStyle{P}_1,\mathcal{J}_1,\mathcal{C}_1 \rangle$ and $\mathcal{L}_2 := \langle \setStyle{A}_2,\setStyle{D}_2,\setStyle{P}_2,\mathcal{J}_2,\mathcal{C}_2 \rangle$ be argument spaces. A function $h:\setStyle{A}_1+\mathcal{J}_1+\mathcal{C}_1 \to \setStyle{A}_2+\mathcal{J}_2+\mathcal{C}_2$ is a homomorphism from $\mathcal{L}_1$ to $\mathcal{L}_2$ iff, for any $\argument{A} \in \setStyle{A}$ and any $j \in \mathcal{J}_1$ and  any $c \in \mathcal{C}_1$, the following equations hold:
%      \[
%      h(j)(h(\argument{A})) = h(j(\argument{A})) \quad \text{ and } \quad  h(c)(h(\argument{A})) = h(c(\argument{A}))
%      \]
% \end{Definition}

Suppose one has a logic and a notion of argument for that logic; for example, IPL and natural deduction. The setup of argument spaces does not yet allow us to relate the two. As in the case-study of Section~\ref{sec:ptv-IPL}, we require a function that determines what sequents are witness by a certain argument in the space.

\begin{Definition}[Ergo]
      Let $ \system{A}$ be an argument space with arguments $\setStyle{A}$. An ergo is a map from arguments to sequents, $\ergo:\setStyle{A} \to \setStyle{S}$.
\end{Definition}
\begin{Definition}[Logical Argument Space]
     A logical argument space (LAS) is a pair $\mathcal{L}:=\langle \system{A}, \ergo \rangle$ in which $\system{A}$ is an argument space and $\ergo$ is an ergo.
\end{Definition}

Observe that it is the use of an ergo that turns proof-theoretic validity from a \emph{semantics of proofs} into a \emph{semantics in terms of proofs}. In Section~\ref{sec:ptv-IPL}, we saw that a natural deduction argument with open assumptions $\Gamma$ and conclusion $\phi$ has the consequence $\Gamma \seq \phi$; this describes an ergo. The notion of validity of arguments renders a LAS a characterize of \emph{some} logic; namely, the logic whose consequence relation consists of all those sequents admitting valid arguments. Given a LAS $\mathcal{L} = \langle \system{A}, \ergo \rangle$, we write $\proves_\mathcal{L}$ to denote the consequence relation of the logic it induces --- that is,
\[
\proves_\mathcal{L} s \qquad \mbox{iff} \qquad \mbox{there is an $\system{A}$-valid argument $\mathcal{A}$ such that $\ergo(\mathcal{A})=s$}
\]

Recall that we take logics to be characterized by sequent calculi, generally conceived (see Section~\ref{sec:logic}). The relationship between the proof-theoretic semantics and the logic is then captured by standard soundness and completeness conditions:  

\begin{Definition}[Soundness and Completeness of Sequent Calculi] \label{def:snc:calculus}
    Let $ \mathcal{L}$ be a LAS and let $\calculus{L}$ be a sequent calculus over sequents $\setStyle{S}$.
    \begin{enumerate}
        \item[-] The calculus $\calculus{L}$ is \emph{sound} for $\mathcal{L}$ iff, for any sequent $s \in \setStyle{S}$, if $\proves_\calculus{L} s$, then $\proves_\mathcal{L} s$.
        \item[-] The calculus $\calculus{L}$ is \emph{complete} for $\mathcal{L}$ iff for any sequent $s \in \setStyle{S}$, if $\proves_\mathcal{L} s$, then $\proves_\calculus{L} s$. 
    \end{enumerate}     
\end{Definition}

Indeed, Definition~\ref{def:snc:calculus} is just an instance of Definition~\ref{def:snc} by taking the logic in the latter as the one defined by the LAS.

This relates logics to argument spaces in general. It remains to relate tactics to argument spaces. This addresses the question at the end of Section~\ref{sec:tactics:IPL}, what argument does a tactic represent? We have an interpretation from a system of tactics to a space of arguments.

\begin{Definition}[Interpretation] \label{def:interpretation}
 Call a functions from $\setStyle{A}$ to $\textsc{LIST}(\setStyle{A})$ an \emph{abstract reduction operator} (ARO). 

An interpretation of  is a function $\llbracket - \rrbracket$ that maps goals to arguments, tactics to AROs, and tacticals to functions from AROs to AROs,  
such that $\tau: G \mapsto \langle [G_1,...,G_n], \pi \rangle$, then $\llbracket \tau \rrbracket (\llbracket G \rrbracket) \mapsto  [\llbracket G_1 \rrbracket,...,\llbracket G_n \rrbracket]$.
\end{Definition}
\begin{Example}
Let $\tau_\land$ and $\tau_\to$ and $\fatsemi$ be as in Section \ref{sec:tactics:IPL}. To this setup we add the interpretation $\llbracket - \rrbracket$ which answers the questions of what arguments are represented by what tactics.

The interpretation $\llbracket - \rrbracket$ acts on goals (i.e., IPL sequents) $\Gamma \seq\phi$ by mapping them to the natural deduction argument consisting of nodes of formulas from $\Gamma$ going directly to a node for $\phi$. It maps tactics to their actions on arguments: For example,
\[
\llbracket \tau_\to \rrbracket(\phi \to \psi) \mapsto\raisebox{-1ex}{\deduce{\psi}{\phi}}
\]
The tactical $\fatsemi$ interprets composition of rules, respecting also discharge. Thus, we have the overall composite action:
\[
\llbracket \tau_\land\fatsemi\tau_\to \rrbracket(\chi\land (\phi \to \psi)) = \raisebox{-3ex}{
 \infer{\chi\land (\phi \to \psi)}{\chi & \infer{\phi \to \psi}{\deduce{\psi}{[\phi]}}} }
\]
\end{Example}

This completes the framework of this paper. We have arguments and logics and tactics, which are pairwise connected in simple, intuitive ways, faithful to mathematical practices, but presented generally. In the next section, we demonstrate that validity, tactical proof, and consequence are coherent throughout the framework.

\subsection{The Correctness Theorems} \label{sec:main:correctness}

We have thus presented a tripartite framework for logic as a mathematics of reasoning: arguments, sequent calculi, and tactics. Above we defined their relationships. It is summarized in the following diagram:

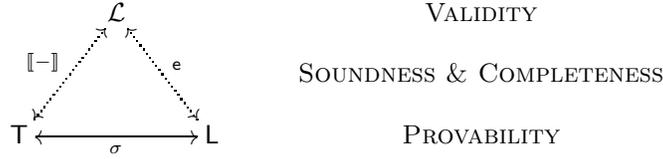
\[
\xymatrix@R=0.5ex{
& \mathcal{L} &  & \textsc{Validity} \\
& & & \\
& & & \textsc{Soundness \& Completeness} \\ 
& & & \\
\calculus{T} \ar@{<..>}[uuuur]^{\llbracket-\rrbracket} \ar@{<->}[rr]_{\sigma} & & \calculus{L} \ar@{<..>}[uuuul]_{\ergo} & \textsc{Provability}
}
\]
Heuristically, tactics $\calculus{T}$ represent (through an interpretation $\llbracket-\rrbracket)$ the constructions of arguments $\mathcal{L}$ that assert (through an ergo $\ergo$) sequents of a logic and that the reasoning steps involved are justified (through a synthesizer $\sigma)$ by the rules of a sequent calculus $\calculus{L}$.

By fixing a sequent calculus, we declare a notion of inference for a logic. This notion of inference justifies a system of tactics if there is a synthesizer. The following coherence result captures this:

\begin{Theorem}\label{thm:coherence}
Let $\calculus{L}$ be a sequent calculus; let $\mathcal{L}=\langle \system{A}, \ergo \rangle$ be a LAS; and let $\calculus{T}$ be a tactical system with achievement $\alpha$. Let $\llbracket - \rrbracket$ be an interpretation of $\calculus{T}$ in $\system{A}$ and let $\propto$ satisfy the following coherence condition: 
\[
\ergo\llbracket G \rrbracket \propto G
\]
Let be $\alpha$ be an $\calculus{L}$-synthesizer for $\calculus{T}$, then the application of a tactic corresponds precisely to an inference in $\calculus{L}$ --- that is, if $\llbracket \tau \rrbracket:\llbracket G \rrbracket \mapsto [\llbracket G_1 \rrbracket ,...,\llbracket G_n \rrbracket]$, then there is an $\calculus{L}$-rule witnessing the following:
\[
\infer[\pi]{\ergo \llbracket G \rrbracket}{\ergo\llbracket G_1 \rrbracket & \hdots & \ergo\llbracket G_n \rrbracket }
\]
\end{Theorem}
\begin{proof}
    By Definition~\ref{def:interpretation}, if  $\tau:G \mapsto \langle [G_1,...,G_n], \pi \rangle$, then $\llbracket \tau \rrbracket:\llbracket G \rrbracket \mapsto [\llbracket G_1 \rrbracket ,...,\llbracket G_n \rrbracket]$. By the coherence condition: $\pi:[\ergo\llbracket G_1 \rrbracket,...,\ergo\llbracket G_n \rrbracket] \mapsto \ergo\llbracket G \rrbracket$. Since $\alpha$ is a synthesizer, the result follows from Definition~\ref{def:synthesizer}.
\end{proof}
\begin{Corollary}
    Calculus $\calculus{L}$ is sound for $\mathcal{L}$.
\end{Corollary}
\begin{proof}
    Theorem~\ref{thm:coherence} states that every rule in $\calculus{L}$ is admissible for the logic induced by $\argument{A}$, which is the soundness condition in Definition~\ref{def:snc:calculus}.
\end{proof}

In this way, a sequent calculus characterizes inference, and a tactical system characterizes the construction of arguments.  This means that the notion of inference for a logic can be as rough or as refined as one desires. For example, one may take the trivial sequent calculus for a consequence relation, which has the consequence of the logic as axioms, but then one admits no tactics. Importantly, one has no way of constructing arguments. Though permissible, this situation is quite degenerate. Instead, one may use some notion of argument to inform what inferences are to be permitted.

We have thus shown that the tripartite framework captured by tactical proof and proof-theoretic semantics is coherent in the sense that arguments, sequent calculi, and tactics have the expected relationship. As stated in Section~\ref{sec:intro}, this framework does not arise from doxastic considerations of what these things should be but rather from how they are used in practice. We have so far been led by a heuristic account and now justify it by a series of examples drawn from the literature on logic. 

\subsection{Examples of the Framework} \label{sec:main:example} 

We provide a brief survey of how various proof-search activities in the literature are instances of the framework in this paper. This survey is far from complete and left at a quite high level as it only illustrates the descriptive power of the framework we have presented --- namely, the relationship between proof-theoretic validity of arguments and inference as witnessed through the reductive logic carried by tactical proof. Of course, in addition, there are also the examples of natural deduction (Example~\ref{ex:naturalPtS}) dialogue games (Example~\ref{ex:games}) in Section~\ref{sec:main:ptv}. 

\begin{Example}[Focused Systems] \label{ex:focusing}
    The problem of proof-search is handling the various choices involved, such as, \emph{inter alia}, the choice of a rule to use and the choice of an instance of that rule. This problem motivates the concept of \emph{focused} proof-search, introduced by Andreoli~\cite{Andreoli1992}, where these things are largely determined. We review a typical approach (see, for example, Chaudhuri~\cite{Chaudhuri2016IL,Chaudhuri2019CL} and Gheorghiu and Marin~\cite{Gheo2021Foc}) for studying focusing.

    One begins a sequent calculus $\calculus{L}$ for which one wishes to establish the focusing property (i.e., that the class of focused proofs is complete for the logic). One introduces an augmented version $\calculus{FL}$,  called the focused system, which arises from enriching the original calculus with control structures and introducing $\cut$. In the framework of this paper, we can describe the situation as follows: one has a system of tactics $\calculus{T}$ that is validated by $\calculus{L}$ (i.e., one has a synthesizer from $\calculus{T}$ to $\calculus{L}$) such that a tactic is interpreted as an $\calculus{FL}$-proof. The space of arguments contains all $\calculus{FL}$-proofs, and the justification operators are given by cut-reduction. This is set up such that the canonical proofs then represent \emph{focused} $\calculus{L}$-proofs. 
    \end{Example}

\begin{Example}[Hyper-sequent Calculi]
Reasoning in substructural and modal logics is often difficult because they seemingly do not admit analytic sequent calculi that do not have extra-logical structures (e.g., labels). Many such logics do admit \emph{hyper}-sequent calculi; that is, calculi over finite multisets of sequents --- see, for example, Baaz et al.~\cite{Baaz} and Ciabattoni et al.~\cite{Ciabattoni}. We can use the framework of this paper to describe the relationship of hyper-sequent calculi to the logic.
    
    One has a system of tactics $\calculus{T}$ of goals that are hyper-sequents such that the tactics are interpreted as reductions in the hyper-sequent calculus. These tactics are valid relative to a notion of achievement defined as follows: a consequence of the logic achieves a hyper-seqeunt iff it is among the sequents in the multiset. 
\end{Example}

\begin{Example}[Analytic Tableaux]
 Analytic tableaux give a computationally useful paradigm of proof in logic. It has been extensively treated for modal logic (see, for example, Fitting and Mendelsohn~\cite{fitting2012first})  and has been used to provide a uniform and modular proof theory for the family of bunched logics (see Docherty and Pym~\cite{docherty2018modular,docherty2019bunched}). 
Typically, these systems make use of prefixed signed formulas. The framework of this paper can be used to describe the relationship between a tableaux system and a logic. 

      One has a system of tactics $\calculus{T}$ whose goals are prefix signed formulas and whose tactics represent expansion rules for the system, and tacticals are sequential composition. These tactics are interpreted in the space of arguments containing the tableaux, in which the canonical proofs are closed tableaux, possibly satisfying a particular expansion scheme. The tactical system is valid relative to a notion of inference supplied by a \emph{relational} sequent calculus in the form of Negri~\cite{Negri2005} and Gheorghiu and Pym~\cite{GRvAC}.
\end{Example}

\section{Conclusion} \label{sec:conclusion}

Typically, when logics are used as a reasoning technology in informatics and other systems-oriented sciences --- for example, the use of logic in program and system verification, natural language processing, knowledge representation, and proof assistants --- it is through the paradigm of reductive logic. That is, when logic is employed in these fields, it is by reducing a putative conclusion to some sufficient premisses (as opposed to the deductive view in which a logic is understood as stating that a conclusion may be inferred from some established premisses).

We declare a notion of inference by fixing a sequent calculus characterizing that logic. What one constructs during backward reasoning is a certificate which we call an \emph{argument}. In general, these arguments are abstract mathematical entities with dynamic and static properties that can be studied in their own right. How is backward reasoning and inference in the logic related? This question is, of course, answered in each use-case; for example, in Section~\ref{sec:ptv-IPL}, we saw how (the construction of) a natural deduction argument can be understood in terms of a sequent calculus. In this paper, we provide a meta-theoretic framework that \emph{uniformly} describes these answers, thereby showing how different deployments of logic as a reasoning tool exercise the same underlying principles.

The essential criterion for whatever notion of \emph{validity} one has for a space of arguments is that the valid arguments are precisely those that certify the consequences of the logic. However, we are concerned about constructing arguments according to backward inference. From this perspective, an argument is valid because it respects a notion of inference for the logic. We use proof-theoretic semantics in the Dummett-Prawitz tradition, grounded in inferentialism, to give a formal definition of soundness and completeness that captures these ideas.

The theory of tactical proof introduced by Milner~\cite{milner1984tactics} is a widely deployed metalogical framework supporting reductive logic. We use tactics to represent the construction of arguments. Specifically, each tactics represents a backward inference step and is interpreted as an argument. A tactic is valid (in a sense given by Milner) precisely when they respect the sequent calculus characterizing the logic in question. This way, tactical proof formally and generally relates arguments and inferences within reductive logic.

In summary, we have a tripartite framework for describing the use of logic as a reasoning tool: there are arguments, sequent calculi, and tactics, and they are related through various validity conditions. This descriptive framework is correct in the sense that these validity conditions, which arise independently according to how the aspects typically relate to one another, cohere. Its descriptive power is illustrated by a series of seemingly disparate examples that are represented uniformly.

Within the framework, we can express other semantics of reductive logic it (e.g., the game semantics for proof-search in IPL --- see Example~\ref{ex:games}), but it remains to conduct such a study in general. Moreover, proof-search arises from reductive logic by introducing control structures that determine what reduction is made at what time. While the framework can witness specific control structures in specific systems (see, for example, the treatment of focused systems in Example~\ref{ex:focusing}), can it give an abstract account of what control is? These questions remain to be studied.

%\emph{Data availability statement:} 

%There is no data associated with this manuscript.

%\input{notes}

\subsection*{Acknowledgements}
We are grateful Tao Gu and the reviewers of an earlier version of this paper for their comments and feedback, which have helped to produce the present paper.

\bibliographystyle{siam}
\bibliography{bib}

\end{document}